\documentclass[runningheads]{llncs}
\usepackage{graphicx}
\usepackage{epsfig}
\usepackage{graphics}
\usepackage{array}
\usepackage{amsmath}
\usepackage{algorithm}
\usepackage{algorithmic}
\usepackage{comment}
\usepackage{fullpage}
\date{}

\title{Spanning Tree Enumeration in 2-trees: Sequential and Parallel Perspective}
\author{Vandhana.C${^1}$ and S.Hima Bindhu${^2}$ and P.Renjith${^3}$ and N.Sadagopan${^3}$ and B.Supraja${^2}$} 
\institute{${^1}$ Department of Information Technology, National Institute of Technology Karnataka, Surathkal, India.\\ ${^2}$ Department of Information Science and Technology, College of Engineering, Guindy, Anna University, India.\\
${^3}$ Indian Institute of Information Technology, Design and Manufacturing, Kancheepuram, Chennai, India. \\
\email{sadagopan@iiitdm.ac.in}}
\begin{document}
\maketitle
\begin{abstract}
For a connected graph, a vertex separator is a set of vertices whose removal creates at least two components. A vertex separator $S$ is minimal if it contains no other separator as a strict subset and a minimum vertex separator is a minimal vertex separator of least cardinality.  A {\em clique} is a set of mutually adjacent vertices.  A 2-tree is a connected graph in which every maximal clique is of size three and every minimal vertex separator is of size two.  A spanning tree of a graph $G$ is a connected and an acyclic subgraph of $G$.  In this paper, we focus our attention on two enumeration problems, both from sequential and parallel perspective.  In particular, we consider listing all possible spanning trees of a 2-tree and listing all perfect elimination orderings of a chordal graph.  As far as enumeration of spanning trees is concerned, our approach is incremental in nature and towards this end, we work with the construction order of the 2-tree, i.e. enumeration of $n$-vertex trees are from $n-1$ vertex trees, $n \geq 4$. Further, we also present a parallel algorithm for spanning tree enumeration using $O(2^n)$ processors.  To our knowledge, this paper makes the first attempt in designing a parallel algorithm for this problem.  We conclude this paper by presenting a sequential and parallel algorithm for enumerating all Perfect Elimination Orderings of a chordal graph. 
\end{abstract}
\section{Introduction}
Enumeration of sets satisfying specific properties is a fascinating problem in the field of computing and it has attracted many researchers in the past.  Properties looked at in the literature are spanning trees \cite{uno,read}, minimal vertex separators \cite{kloks}, cycles \cite{read}, maximal independent sets \cite{papadi}, etc.  Interestingly, these problems find applications in Computer Networks and Circuit Analysis \cite{read,minty}.  In this paper, we revisit enumeration of spanning trees restricted to 2-trees.  The algorithms available in the literature either follow back-tracking approach to list all spanning trees or list spanning trees using fundamental cycles \cite{read,gabow}.  As far as enumeration problems are concerned, asking for a polynomial-time algorithm to enumerate all desired sets is quite unlikely as the number of such sets is exponential in the input size.  Interestingly, the results reported in the literature list all spanning trees with polynomial delay between consecutive spanning trees and this is the best possible for enumeration problems. Having highlighted the inherent difficulty of enumeration problems, a natural approach to speed up enumeration is to design parallel algorithms wherein more than one tree is listed at a time.  Although parallel algorithms have received much attention in the past for other combinatorial problems such as sorting \cite{kruskal}, planarity testing \cite{planar}, connectivity augmentation \cite{hsu}, surprisingly, no parallel algorithm exists for enumeration problems.  Most importantly, effective parallelism can be achieved for enumeration problems compared to other combinatorial optimization problems with the help of parallel algorithmics.  To the best of our knowledge, this paper presents the first parallel algorithm for enumeration of spanning trees in 2-trees.  We first present a new sequential algorithm for listing all spanning trees of a 2-tree. Our novel approach is iterative in nature in which spanning trees of $n$-vertex 2-tree is generated  using spanning trees of $n-1$ vertex 2-tree.  This approach is fundamentally different from the results reported in \cite{read,gabow}.  Most importantly, the overall structure of our sequential algorithm can be easily extended to design parallel algorithm for listing all spanning trees of a 2-tree.  In particular, using CREW PRAM model and with the help of $O(2^n)$ processors, we present a parallel algorithm to list all spanning trees in a 2-tree. We actually present a framework and we believe that this can be extended to $k$-trees, chordal graphs and other graphs which have vertex elimination orderings.  Our sequential approach looks at enumeration as a two stage process wherein stage-1 enumerates all spanning trees of $n$ vertex 2-tree with $n^{th}$ vertex being a leaf and in stage-2, it generates all spanning trees where $n^{th}$ vertex is a non-leaf.  We also highlight that this two stage process naturally yields tight parallelism and we believe that this is the main contribution of this paper. Further, each processor incurs $O(n)$ time, linear in the input size to output a tree. As far as bounds are concerned, Cayley's formula \cite{west} shows that the number of spanning trees of an $n$-vertex graph is upper bounded by $n^{n-2}$ and this bound is tight for complete graphs.  Alternately, we can also get the number using Kirchoff's Matrix Tree theorem \cite{west}.  In this paper, for 2-trees we present a recurrence relation to capture the number of spanning trees of an $n$-vertex 2-tree.  Our initial motivation was to check whether 2-trees has polynomial number of spanning trees; however, we show that there are $\Omega(2^n)$ spanning trees in any $n$-vertex 2-tree. Moreover, it is this bound that helped us to fix the number of processors while designing parallel algorithms. \\
{\bf Road Map:} In Section \ref{sequential}, we present a new sequential algorithm to list all spanning trees of a 2-tree.  Our two stage algorithm along with combinatorics is presented in Section \ref{sequential}.  In Section \ref{parallel}, we design and analyze parallel algorithm for spanning tree enumeration restricted to 2-trees. \\
\subsection{Graph-theoretic Preliminaries}
Notation and definitions are as per \cite{west,golu}.  Let $G =(V,E)$ be an undirected connected graph where $V(G)$ is the set of vertices and $E(G) \subseteq \{\{u,v\}~|~ u,v \in V(G)$, $u \not= v \}$. For $v \in V(G)$, $N_G(v)=\{u~|~ \{u,v\} \in E(G)\}$ and $d_G(v)=|N_G(v)|$ refers to the degree of $v$ in $G$.  For $S \subset V(G)$, $G[S]$ denotes the graph induced on the set $S$ and $G \setminus S$ is the induced graph on the vertex set $V(G) \setminus S$.  A vertex separator of a graph $G$ is a set $S \subseteq V(G)$ such that $G \setminus S$ has more than one connected component and $S$ is minimal if there does not exist $S' \subset S$ such that $S'$ is also a vertex separator.  A minimum vertex separator is a minimal vertex separator of least size.  A cycle is a connected graph in which the degree of each vertex is two.  A tree is a connected and an acyclic graph.  A set $S \subseteq V(G)$ is a clique if for all $u,v \in S, \{u,v\} \in E(G)$ and $S$ is maximal if there is no clique $S' \supset S$.  A graph $G$ is a 2-tree if every maximal clique in $G$ is of size three and every minimal vertex separator in $G$ is of size two.  2-trees can be defined iteratively as follows: A clique on 3-vertices (3-clique) is a 2-tree and if $H$ is a 2-tree, then the graph $H'=H + v$ obtained from $H$ by adding $v$ such that $N_{H'}(v)$ is an edge (2-clique) in $H$ is also a 2-tree.  An example is shown in Figure \ref{2tree}.  A 2-tree $G$ also has a vertex elimination ordering which is an ordering $(v_1,v_2,\ldots,v_n)$ such that for any $v_i$, $N_H(v_i)$ in the subgraph $H$ induced on the set $\{v_i,v_{i+1},\ldots,v_n\}$ is a 2-clique. We call such $v_i$ as 2-simplicial and the associated ordering, a 2-simplicial ordering of $G$.  This is a special case of perfect elimination ordering (PEO) which is an ordering $(v_1,v_2,\ldots,v_n)$ such that for any $v_i$, $N_H(v_i)$ in the subgraph $H$ induced on the set $\{v_i,v_{i+1},\ldots,v_n\}$ is a clique. A graph is chordal if every cycle $C$ of length at least 4 has a chord in $C$, an edge joining a pair of non-consecutive vertices in $C$ of $G$.  It is well known that chordal graphs have perfect elimination ordering. We call $N_H(v_i)$ as the {\em higher neighbourhood} of $v_i$.  Note that for a 2-tree, for any $v_i$, the higher neighbourhood is always a 2-clique (edge). A 2-tree and its 2-simplicial ordering is illustrated in Figure \ref{2tree}.
\begin{figure}
\begin{center}
\includegraphics[scale=0.4]{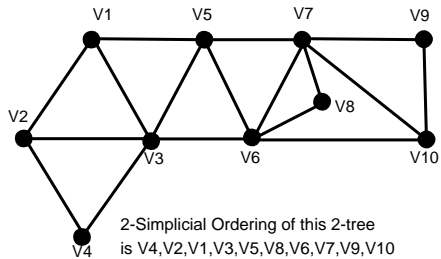} 
\caption{A 2-tree and its 2-simplicial ordering} \label{2tree}
\end{center}
\vspace{-0.5cm}
\end{figure}
\subsection{Parallel Computing Preliminaries}
In this paper, we work with Parallel Random Access Machine (PRAM) Model \cite{JaJa}.  It consists of a set of $n$ processors all connected to a shared memory.  The time complexity of a parallel algorithm is measured using $O($number of processors $\times$ time for each processor$)$.  This is also known as processor-time product.  Access policy must be enforced when two processors are trying to Read/Write into a cell. This can be resolved using one of the following strategies:
\begin{itemize}
\item Exclusive Read and Exclusive Write (EREW): Only one processor is allowed to read/write into a cell
\item Concurrent Read and Exclusive Write (CREW): More than one processor can read a cell but only one is allowed to write at a time
\item Concurrent Read and Concurrent Write (CRCW): All processors can read and write into a cell at a time.
\end{itemize}
In our work, we restrict our attention to CREW PRAM model.  For a problem $Q$ with input size $N$ and $p$ processors,  the speed-up is defined as $S_p(N)= \frac{T_1(N)}{T_p(N)}$, where $T_p(N)$ is the time taken by the parallel algorithm on a problem size $N$ with $p$ ($p \geq 2$) processors and $T_1(N)$ is the time taken by the fastest sequential algorithm (in this case $p=1$) to solve $Q$.  The efficiency is defined as $E_p(N)=\frac{S_p(N)}{p}$.
\section{Enumeration of Spanning Trees for a 2-tree: A New Sequential Approach}
\label{sequential}
In this section, we present a new sequential approach for the enumeration of all spanning trees in 2-trees.  Our new sequential algorithm is iterative in nature and has got three phases.  In Phase-1, we obtain a vertex elimination ordering by removing all the 2-simplicial vertices of the 2-tree one by one until we are left with a 3-clique, which is considered to be our base graph.  This will give us the 2-simplicial ordering.  In Phase-2, we construct the spanning trees for our base graph.  Since our base graph is a 3-clique, we will get three spanning trees.  Using these three spanning trees, and with the help of 2-simplicial ordering, in Phase-3, we construct all the other spanning trees iteratively. In particular, we first consider the last four vertices of the ordering, namely, $\{v_4,v_3,v_2,v_1\}$ and enumerate all spanning trees for this set.  Using this set and in the next iteration, we generate all spanning trees of the set $\{v_5,v_4,\ldots,v_1\}$ and so on. In Phase-3, at the beginning of $i^{th}$ iteration, we have all spanning trees of the 2-tree induced on the set $\{v_{i-1},\ldots,v_1\}$ and using this set, we first generate spanning trees of $i$-vertex 2-tree with $v_i$ being a leaf, followed by generation of spanning trees where $v_i$ is a non-leaf.  An illustration is given in Figure \ref{algotrace}. \\
\noindent
{\bf Trace of Algorithm \ref{mainalgo}:} {\tt Fig:1a} of Figure \ref{algotrace} shows a 2-tree on 5-vertices.  The base 2-tree consists of 3-vertices and there are 3-different spanning trees for the base graph which are shown in {\tt Fig: 2a-2c}; this task is done by line 5 of Algorithm \ref{mainalgo}. Lines 8-11 of Algorithm \ref{mainalgo} generate the sequence of spanning trees with $v_4$ as a leaf, which are illustrated in  {\tt Fig: 3a-3f}.  Subsequently, spanning trees with $v_4$ as a non-leaf are generated by lines 12-18, which are illustrated in {\tt Fig: 4a-4b}.  In the next iteration, we first generate spanning trees with $v_5$ as a leaf, which are shown in {\tt Fig: 5a-5p}.  Finally, {\tt Fig: 6a-6e} illustrate the set of spanning trees in which $v_5$ is a non-leaf. 
\begin{algorithm}
\caption{2-simplicial ordering of a 2-tree {\em 2-simplicial-ordering ($G$)}}
\label{ordering}
\begin{algorithmic}[1] 
\STATE{Let $S= \phi$ and $n = |V(G)|$. {\tt /* A set to maintain ordering of vertices in $G$ */}}
\WHILE{$|V(G)| >=4$}
\STATE{Identify a vertex $u_n$ of degree two in $G$ and add $u_n$ to $S$, $S=S \cup \{u_n\}$}
\STATE{Remove $u_n$ from $G$, i.e., $V(G)=V(G) \setminus \{u_n\}$, $E(G)=E(G)\setminus \{\{u_n,x\},\{u_n,y\}\}$ where $x,y \in N_G(u_n)$. $n \leftarrow n-1$}
\ENDWHILE
\STATE{$S=S \cup \{v_3,v_2,v_1\}$ and return $S$}
\end{algorithmic}
\end{algorithm}
\vspace{-1cm}
\begin{algorithm}
\caption{Spanning Trees of a 3-clique (Base Graph of a 2-tree) {\em Spanning-trees-base-graph()}}
\label{basegraph}
\begin{algorithmic}[1]
\STATE{Let $V(H)=\{v1,v2,v3\}$}
\STATE{Construct spanning trees $T_1, T_2, T_3$ such that $V(T_i)=V(H)$ for all $1 \leq i \leq 3$ and $E(T_1)=\{\{v1,v2\},\{v1,v3\}\}$, $E(T_2)=\{\{v1,v2\},\{v2,v3\}\}$, $E(T_3)=\{\{v1,v3\},\{v2,v3\}\}$ }
\STATE{Augment $T_1,T_2,T_3$ to $ENUM$ {\tt /* $ENUM$ contains the set of spanning trees */}}
\end{algorithmic}
\end{algorithm}
\begin{algorithm}
\caption{Spanning Tree Enumeration in 2-trees {\em Enumerate-Spanning-trees($G$)}}
\label{mainalgo}
\begin{algorithmic}[1]
\STATE{{\tt Input:} A 2-tree $G$}
\STATE{{\tt Output:} List all spanning trees of $G$. {\tt /* The set $ENUM$ contains all spanning trees of $G$  */}}
\STATE{Get an ordering of vertices in G; {\em 2-simplicial-ordering($G$).}}
\STATE{Let $\{v_n,v_{n-1}, \ldots, v3,v2,v1\}$ be the 2-simplicial ordering.}
\STATE{Call {\em Spanning-trees-base-graph()} to get spanning trees for the base graph.}
\FOR{$i$ = $4$ to $n$}
\STATE{Let $\{x,y\}$ be the higher neighbourhood of $v_i$. }
\FOR{each tree $T$ in $ENUM$}
\STATE{Add the edge $\{v_i,x\}$ to $T$ to get a tree $T_1$ and
add the edge $\{v_i,y\}$ to $T$ to get a tree $T_2$.}
\STATE{Augment $T_1,T_2$ to $ENUM$}
\ENDFOR
\FOR{each tree $T$ in $ENUM$ such that $\{x,y\} \notin E(T)$}
\STATE{Add $v_i$ to $T$ to get a graph $H$ i.e., $V(H)= V(T) \cup \{v_i\}$ and $E(H)=E(T)\cup \{\{v_i,x\},\{v_i,y\}\}$}
\STATE{Let $C$ be the cycle in $H$ of length $k \geq 4$ with $V(C)=\{w_1=v_i,w_2=x,w_3,\ldots,w_k=y\}$ }
\FOR{each edge $\{w_i,w_j\}$ in $C$ such that $i \neq 1$,$j \neq 1$}
\STATE{Delete $\{w_i,w_j\}$ to get a tree $T'$ from $H$ and augment $T'$ to $ENUM$ }
\ENDFOR
\ENDFOR
\ENDFOR
\end{algorithmic}
\end{algorithm}
\begin{figure}
\vspace{-0.5cm}
\begin{center}
\includegraphics[scale=0.5]{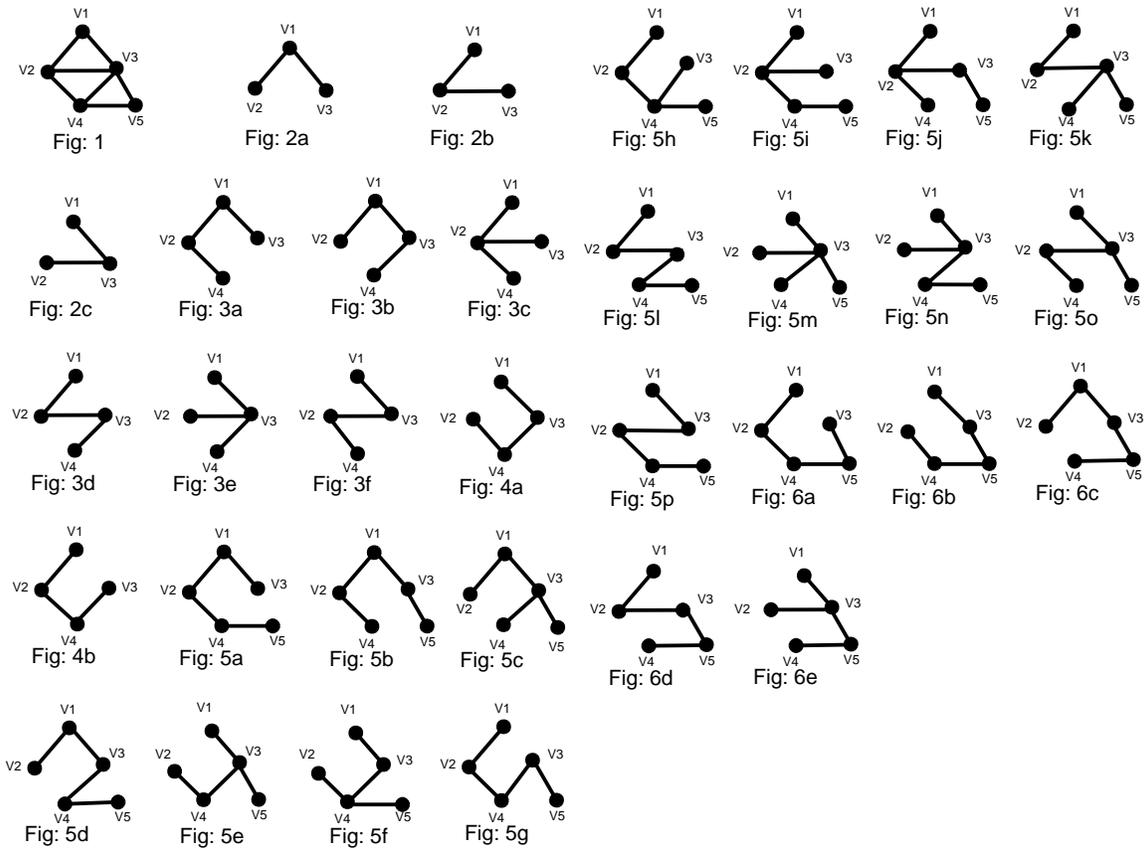} 
\caption{Trace of Spanning Tree Enumeration Algorithm}\label{algotrace}
\end{center}
\vspace{-0.8cm}
\end{figure}
\begin{theorem}
For a 2-tree $G$, Algorithm \ref{mainalgo} lists all spanning trees of $G$.
\end{theorem}
\begin{proof}
We present a proof by mathematical induction on $|V(G)|$. {\em Base:} $|V(G)|=3$.  The only 2-tree on 3-vertices is a clique on 3-vertices.  There are exactly three spanning trees and Algorithm \ref{basegraph} correctly outputs all three spanning trees of $G$.  Therefore, the claim is true for base case. {\em Hypothesis:} Assume that our claim is true for 2-trees with less than $n$ vertices ($n \geq 3$).  {\em Induction Step:} Let $G$ be a 2-tree on $n \geq 4$ vertices.  Consider the 2-simplicial ordering $(v_n,\ldots,v_1)$ of $G$.  Clearly the higher neighbourhood $\{x,y\}$ of $v_n$ is a 2-clique.  Consider the graph $G'$ obtained from $G$ by removing $v_n$. i.e., $V(G')=V(G) \setminus \{v_n\}$ and $E(G')=E(G) \setminus \{\{v_n,x\},\{v_n,y\}\}$.  Clearly $G'$ has less than $n$ vertices and by our induction hypothesis, Algorithm \ref{mainalgo} correctly outputs all spanning trees of $G'$.  We now argue that our algorithm indeed enumerates all spanning trees for $G$ as well.  We observe that in the set $ENUM$ (the set of all spanning trees of $G$), $v_n$ appears as a leaf or a non-leaf.  Based on this observation, we consider two cases to complete our induction. {\em Case 1:} $v_n$ is a leaf.  Note that the lines 8-11 of Algorithm \ref{mainalgo} considers each spanning tree on $n-1$ vertices ($n \geq 4$) and augments either the edge $\{v_n,x\}$ or $\{v_n,y\}$ to get a spanning tree on $n$ vertices ($n \geq 4$) and in either augmentation, $v_n$ is a leaf node. {\em Case 2:} $v_n$ is a non-leaf.  Lines 12-17 of Algorithm \ref{mainalgo} handle Case 2.  For each spanning tree $T$ in $ENUM$ where $\{x,y\} \notin E(T)$, we augment $v_n$ to $T$.  This creates a cycle $C$ in the associated graph and to get a spanning tree with $v_n$ as a non-leaf, Algorithm \ref{mainalgo} removes an edge from $C$ other than the edges incident on $v_n$ to get a new spanning tree of $G$.  To ensure that $v_n$ is non-leaf, the edges incident on $v_n$ are not removed by the algorithm.  Moreover, there are $|C|-2$ more spanning trees constructed out of $T$ with $v_n$ as the non-leaf. Further, the above observation is true for each $T$ in $ENUM$.  Lines 12-17 of the algorithm identifies all such spanning trees by considering each $T$ in $ENUM$.  Note that in $ENUM$, $T$ with $\{x,y\} \in E(G)$ is not considered for discussion as the spanning trees generated by $T$ is taken care by Case 2 itself,  i.e. when $v_n$ is augmented to such $T$, it creates a cycle of length three with the set $\{v_n,x,y\}$ and to get a spanning tree $T'$ with $v_n$ as the non-leaf, remove the edge $\{x,y\}$.  We observe that the spanning tree $T'$ with $v_n$ as non-leaf is already generated by our algorithm during Case 2.  This completes our induction and therefore, our claim follows. \qed
\end{proof}
\subsection{Bounds on the Number of Spanning Trees}
Our initial motivation was to study whether 2-trees have polynomial number of spanning trees.  Interestingly, we observed that in any 2-tree, the number of spanning trees is $\Omega(2^n)$.   Further, we will also present the exact bound using recurrence relations.  We also highlight that the spanning trees generated by lines 8-11 of our algorithm are without repetition and repetition is due to lines 12-17 of our algorithm.  Nevertheless, for any 2-tree, all spanning trees are generated by Algorithm \ref{mainalgo}.  The lower bound presented in the next observation helped us in fixing the number of processors in parallel algorithmics.  
\begin{theorem}
\label{lb}
Let $G$ be a 2-tree. The number of spanning trees is $\Omega(2^n)$.
\end{theorem}
\begin{proof}
We present a proof using counting argument.  Let $T(n)$ denote the number of spanning trees on $n$-vertex 2-trees.  Note that $T(3)=3$, which is the number of spanning trees of a 3-clique.  At every iteration $i$, $4 \leq i \leq n$, the vertex $v_i$ is adjacent to either $x$ or $y$, where $\{x,y\}$ is the higher neighbourhood of $v_i$ with respect to the 2-simplicial ordering.  This implies that for each spanning tree on $(i-1)$ vertices, the addition of $v_i$ creates at least two more spanning trees and in particular, two spanning trees with $v_i$ as the leaf node.  Therefore, we conclude that $T(n) \geq 2 \times T(n-1)$, $n \geq 4$ and $T(3)=3$.  Solving this recurrence relation, we get $T(n) \geq \frac{3}{8} \times 2^n = \Omega(2^n)$.  \qed
\end{proof}
The next theorem presents an upper bound on the number of spanning trees generated by our algorithm.  
\begin{theorem}
For a 2-tree $G$, the number of spanning trees is $T(n) \leq 2.T(n-1) + (|C_n| - 2).T(n-1)$, where $|C_n|$ denotes the length of the cycle created during iteration $n$ and $n \geq 4$, $T(3)=3$. 
\end{theorem}
\begin{proof}
From Theorem \ref{lb}, we know that the number of spanning trees in which $n^{th}$ vertex appears as a leaf is $2.T(n-1)$.  
Note that during $i^{th}$ iteration, when $v_i$ is augmented as a non-leaf to a tree $T$ in $ENUM$, a cycle $C_i$ is created and the removal of any edge in $C_i$ other than the edges incident on $v_i$ creates a new spanning tree.  This shows that for each spanning tree on $(i-1)$ vertices in the set $ENUM$, our algorithm constructs $|C_i|-2$ more spanning trees on $i$ vertices.  From the previous theorem, we know that there are $2.T(n-1)$ spanning trees with $v_n$ as a leaf. Therefore, the claim follows.\qed
\end{proof}
\subsection{Implementation and Run-time Analysis}
\label{implementation}
In this section, we describe the data structures used to implement our algorithm and using which we also analyze the time complexity of our algorithm.  From the input 2-tree, the first task is to get a 2-simplicial ordering.  Towards this end, we maintain a hash table $H1$ indexed by vertex label and to the cell $v_i$ corresponding to vertex $v_i$ in $H1$, we store the higher neighbourhood of $v_i$.  We also maintain a stack $S1$ to store the 2-simplicial ordering.  We populate both $S1$ and $H1$ during Algorithm \ref{ordering}; while removing a 2-simplicial vertex $v$, it is pushed into $S1$ and its higher neighbourhood is stored in $H1$.  Identification of 2-simplicial vertex can be done in $O(n)$ time and population of $S1$ and $H1$ incurs $O(1)$ time.  Therefore, the overall effort for $n$-iterations is $O(n^2)$.  At iteration $i$, the top of $S1$ contains the vertex $v_i$ which can be fetched in $O(1)$ time.\\
We maintain two dynamic queues $Q1$ and $Q2$ to keep track of the trees generated.  Both queues contain a list of pointers and each pointer points to a tree.  Initially $Q1$ is populated with three pointers corresponding to spanning trees of the base graph and each pointer points to a list of edges of the respective tree.  i.e., $Q1$ contains three pointers namely $T1,T2,T3$ and $T1$ points to a list $\{e_1,e_2\}$ of edges.  We only maintain a list of edge labels with each tree pointer.   Spanning trees for the next iteration are generated using these three spanning trees and newly created spanning trees are stored in $Q2$.  $Q2$ will also contain a list of pointers and each tree pointer points to the list of edge labels of the tree.  Using $Q2$, we generate the next set of spanning trees and that would be stored in $Q1$.  In each iteration, we make use of $Q1(Q2)$ to generate the next set of spanning trees which would be stored in $Q2(Q1)$.  At the end, either $Q1(Q2)$ contains the set of spanning trees of the given 2-tree. \\
During $i^{th}$ iteration, for a tree $T$ in $Q1(Q2)$, creating a new tree from $T$ with $v_i$ as a leaf incurs $O(n)$ time.  This is true because, using $H1$, we can get the higher neighbourhood of $v_i$ in constant time.  Since the size of the higher neighbourhood is two, $T$ creates two more  spanning trees with $v_i$ as a leaf.  Since we are adding an edge label to the already existing list of $T$, this task requires $O(n)$ time and is pushed into $Q2(Q1)$.  While creating a new tree from $T$ with $v_i$ as a non-leaf, Algorithm \ref{mainalgo} first checks whether the higher neighbourhood of $v_i$ is non-adjacent in $T$, which can be done in $O(n)$-time.  Further, Algorithm \ref{mainalgo} creates a graph $H$ in which there exists a cycle with $v_i$. Cycle detection can be done using standard Depth First Search algorithm in $O(n)$ time.  Moreover, we get $|C|-2$ more spanning trees from $T$ which would be pushed into $Q2(Q1)$, where $|C|$ denotes the length of cycle in the graph $H$.  The above task can be done in $O(n)$ time.  Therefore, the overall time complexity of our algorithm is $O(|ENUM|.n)$, where $ENUM$ is the set of spanning trees generated by our algorithm for a 2-tree. 
\section{Parallel Algorithm for Spanning Tree Enumeration in 2-trees}
\label{parallel}
The overall structure of our sequential algorithm naturally gives a parallel algorithm to enumerate all spanning trees in 2-trees. Towards this end, we make use of $O(2^n)$ processors and our implementation is based on CREW PRAM.  We first generate 2-simplicial ordering of a 2-tree using $O(n)$ processors.  Each processor incurs $O(1)$ time to check whether a vertex is simplicial or not.  Following this, we generate all spanning trees iteratively starting from the base graph.  In each iteration $i$, we make use of $O(2^i)$ processors and each processor incurs $O(n)$ time in this enumeration process.   The implementation of Algorithm \ref{parallelenumeration} is similar to the implementation discussed in Section \ref{implementation} and using which, we see that each processor incurs $O(n)$ time effort in the enumeration process.  Overall, there are $O(2^n)$ processors being used by our algorithm.  
\begin{algorithm}
\caption{Parallel Algorithm to generate 2-simplicial ordering of a 2-tree {\em parallel-2-simplicial-ordering(G)}}
\label{parallelsimpordering}
\begin{algorithmic}[1]
\STATE{{\tt /* Use $O(n)$ processors to output 2-simplicial ordering */}}
\WHILE{$|V(G)| \geq 4$}
\STATE{Let $\{P_1,\ldots,P_x\}$ be the set of processors, $x=|V(G)|$ and $x=O(n)$. The set $S$ contains the ordering.}
\STATE{$P_i$ checks whether $v_i$ is a simplicial vertex or not.  If $v_i$ is simplicial, then $P_i$ adds $v_i$ to $S$ and remove $v_i$ from $G$.}
\STATE{Also, maintain a table containing higher neighbourhood of $v_i$.}
\ENDWHILE
\STATE{Using a single processor augment the set $S$ with $\{v_3,v_2,v_1\}$ and Return $S$}
\end{algorithmic}
\end{algorithm}
\begin{algorithm}
\caption{Parallel Enumeration of Spanning Trees in 2-trees {\em Parallel-Enumeration-Spanning-Trees (G)}}
\label{parallelenumeration}
\begin{algorithmic}[1]
\STATE{Get 2-simplicial ordering: {\em parallel-2-simplicial-ordering(G)}}
\STATE{Use one processor to generate spanning trees for the base graph which is a 2-tree on 3-vertices: {\em Spanning-trees-base-graph()} and update the set $ENUM$, the set of spanning trees of $G$}
\FOR{$i$ $=$ $4$ to $n$}
\STATE{{\tt /* Each iteration uses $O(2^i)$ processors */}}
\WHILE{there are spanning trees on $(i-1)$ vertices in $ENUM$}
\STATE{Let $P=\{P_1,\ldots,P_x\}, x=O(2^i)$ be the set of processors such that $P_j$ focuses on the $j^{th}$ spanning tree $T_j$ in $ENUM$ }
\STATE {{\tt /* generate spanning trees with $v_i$ as a leaf */}}
\STATE{$P_j$ adds the edge $\{v_i,x\}$ to get $T^j_1$ and adds the edge $\{v_i,y\}$ to get $T^j_2$ where $\{x,y\}$ is the higher neighbourhood of $v_i$.  Add $T^j_1$ and $T^j_2$ to $ENUM$}
\STATE {{\tt /* generate spanning trees with $v_i$ as a non-leaf */}}
\STATE{Each processor in $P$ identifies a spanning tree $T$ in $ENUM$ such that $\{x,y\} \notin E(T)$, where $\{x,y\}$ is the higher neighbourhood of $v_i$ }
\STATE{Add $v_i$ to $T$ to get a graph $H$ i.e., $V(H)= V(T) \cup \{v_i\}$ and $E(H)=E(T)\cup \{\{v_i,x\},\{v_i,y\}\}$}
\STATE{Let $C$ be the cycle in $H$ of length $k \geq 4$ with $V(C)=\{w_1=v_i,w_2=x,w_3,\ldots,w_k=y\}$ }
\FOR{each edge $\{w_i,w_j\}$ in $C$ such that $i \neq 1$,$j \neq 1$}
\STATE{Delete $\{w_i,w_j\}$ to get a tree $T'$ from $H$ and augment $T'$ to $ENUM$ }
\ENDFOR
\ENDWHILE
\ENDFOR
\end{algorithmic}
\end{algorithm}
\newpage
\section{Related Problem: Enumeration of Perfect Elimination Orderings}
In the earlier section, as part of enumeration process, we presented an algorithm to list 2-simplicial ordering of a 2-tree.   It is easy to observe that for a 2-tree, there is more than one 2-simplicial ordering.  Having looked at enumeration problem in this paper, it is natural to think of enumeration of 2-simplicial ordering of a 2-tree.  In fact, in this section, we look at this question in larger dimension.  Towards this end, we consider listing all perfect elimination orderings of a chordal graph.  We reiterate the fact that 2-trees are a subclass of chordal graphs.
\begin{algorithm}
\caption{Enumeration of Perfect Elimination Orderings of a Chordal graph {\em Enumerate-PEOs(Graph $G$)}}
\label{peoenumeration}
\begin{algorithmic}[1]
\STATE{{\tt Input:} Chordal Graph $G$}
\STATE{{\tt Output:} Enumerate PEOs of $G$}
\IF{$G$ is a clique on $l>0$ vertices}
\STATE{Return all permutations of the set $\{v_1,\ldots,v_l\}$}
\ELSE
\STATE{{\tt /* Recursively remove simplicial vertices till the graph becomes a clique */} }
\FOR{each simplicial vertex $v$ in $G$}
\STATE{Remove $v$ from $G$}
\STATE{Call {\em Enumerate-PEOs(G)}}
\ENDFOR
\STATE{{\tt /* Order in which simplicial vertices are removed along with possible permutations of the base clique yields all PEOs */}}
\ENDIF
\end{algorithmic}
\end{algorithm}
\newpage
\subsection{Bounds on the number of PEOs}
It is well-known that in any non-complete chordal graph, there exist at least two non-adjacent simplicial vertices \cite{golu}.  This shows that the number of PEOs is at least \\
$T(n) \geq 2.T(n-1)$, where $T(n)$ denotes the number of PEOs on $n$-vertex chordal graph, $T(3)=6$ \\
Solving for $T(n)$ gives $T(n)= \Omega (2^n)$. \\
For a complete chordal graph, every vertex is simplicial and therefore, the number of PEOs is just the number of permutations of the vertex set which is $O(n!)$.
\subsection{Parallel Algorithm for PEO Enumeration}
\begin{algorithm}
\caption{Parallel Enumeration of Perfect Elimination Orderings of a Chordal graph {\em Parallel-Enumerate-PEOs(Graph $G$)}}
\label{parallelpeoenumeration}
\begin{algorithmic}[1]
\IF{$G$ is a clique on $l>0$ vertices}
\STATE{Use $O(n)$ processors in parallel and return all permutations of the set $\{v_1,\ldots,v_l\}$.  Processor $i$ generates all permutations with $v_i$ as the first vertex.}
\ELSE
\STATE{{\tt /* Recursively remove simplicial vertices till the graph becomes a clique */}}
\STATE{{\tt /* Use $O(2^n)$ processors in parallel */}}
\FOR{each simplicial vertex $v$ in $G$}
\STATE{Remove $v$ from $G$}
\STATE{Call {\em Enumerate-PEOs(G)}}
\ENDFOR
\STATE{{\tt /* Order in which simplicial vertices are removed along with possible permutations of the base clique yields all PEOs */}}
\ENDIF
\end{algorithmic}
\end{algorithm}
\noindent
{\bf Implementation Details:} Given a chordal graph $G$, both Algorithm \ref{peoenumeration} and Algorithm \ref{parallelpeoenumeration} list all perfect elimination orderings of $G$.  This is true because, in both the algorithms, we first find a set $\{v_1,\ldots,v_k\}, k \geq 2$ of simplicial vertices in $G$ and using which we recursively list all PEOs.  We make use of {\em tree} data structure to store all PEOs.  The root of tree $T$ is labelled with $G$ and its neighbours are $G_1,\ldots,G_k$, where $G_i$ corresponds to the graph obtained from $G$ by removing the simplicial vertex $v_i$ and the edges in $T$ store the labels of simplicial vertices being removed at that recursion.  Similarly, for the node in $T$ corresponding to $G_i$, its neighbours are graphs $\{H_1,\ldots,H_k\}$ where $G_i$ contains $k$ simplicial vertices and $H_j$ corresponds to the graph obtained from $G_i$ by removing the simplicial vertex $v_j$.  When the recursion bottoms out, the paths from the root to leaves precisely give all PEOs.  For parallel algorithm, we make use of $O(2^n)$ processors as the lower bound on the number of PEOs is $\Omega(2^n)$.  
\section{Summary and Directions for Further Research}
In this paper, we have presented a novel approach to enumerate spanning trees of a 2-tree from both sequential and parallel perspective.  Our parallel approach can be implemented using $O(2^n)$ processors.  A natural extension of this approach would be to enumerate spanning trees of chordal graphs using PEO as a tool.  We have also looked at the enumeration of PEOs of chordal graphs both from sequential and parallel perspective.  The iterative approach proposed here naturally yields a  parallel algorithm and we believe that this technique can be used to discover parallel algorithms for other enumeration problems such as maximal independent set, minimal vertex separator, etc.

\end{document}